\documentclass[twocolumn]{article}
\usepackage{txfonts}
\usepackage{amssymb}
\usepackage{graphicx}
\usepackage{dsfont}
\usepackage{stmaryrd}

\usepackage[T1]{fontenc}
\usepackage[utf8]{inputenc}

\usepackage{savesym}
\savesymbol{iint}
\savesymbol{iiint}
\savesymbol{iiiint}
\savesymbol{idotsint}
\savesymbol{openbox}
\usepackage{amsmath}
\usepackage{amscd}

\usepackage{amsthm}

\newtheorem{definition}{Definition}
\newtheorem{proposition}{Proposition}
\newtheorem{theorem}{Theorem}
\newtheorem{lemma}{Lemma}
\newtheorem{remark}{Remark}

\usepackage{subfigure}

\begin{document}

\title{Introducing the truly chaotic finite state machines and theirs applications in security field}
\author{Christophe Guyeux and Qianxue Wang and Xiole Fang and Jacques Bahi}

\maketitle

\abstract
The truly chaotic finite machines introduced by authors in previous research papers are presented here.
A state of the art in this discipline, encompassing all previous mathematical investigations, is 
provided, explaining how finite state machines can behave chaotically regarding the slight
alteration of their inputs. This behavior is explained using Turing machines and formalized 
thanks to a special family of discrete dynamical systems called chaotic iterations. An 
illustrative example is finally given in the field of hash functions.
\endabstract

\section{Introduction}

The use of chaotic dynamics in cryptography is often 
disputed as a finite state machine is reputed to always
enter into a cycle. Even though such a regular  
behavior is not completely opposed to almost all 
definitions of chaos in mathematics, constituting 
a kind of border situation in case of discrete sets, 
this situation appears as problematic to cryptologists
that consider periodic dynamics and chaos as antithetical.
This problem can be solved by introducing truly chaotic
finite machines.
A second situation where such chaotic machines can 
serve is in the numerical simulations of chaotic real
phenomena. Iterating chaotic dynamical systems of the
real line on finite state machines leads to truncated
sequences. It is possible to show that such periodic
orbits of truncated terms are as close as possible to
truly chaotic real ones via the shadow lemma. However
the orbit that is approximated by the finite machine
has a initial condition with a priori no relation with
the targeted one.

Our proposal is to constitute truly chaotic finite
machines, that is, finite machines that can be 
rigorously  proven as chaotic, as defined by 
Devaney~\cite{Devaney}, Knudsen~\cite{Knudsen94}, and so on.
The key idea is to consider that the finite machine
is not separated from the outside world but that it must interact with 
it in order to be useful. At each iteration, the new
input provided to the finite machine can be used 
together with its current state to produce the next
output. By doing so, the finite machine iterates on the
finite cartesian product of its possible states multiply
by all the possible inputs. The machine can be
written as an iterative process, thus it still remains
to study the behavior of outputs on slight modifications
on the inputs.
After having recalled the various existing notions of
chaos in mathematical topology, this idea is formalized
theoretically using Turing machines and explained 
practically thanks to the so-called chaotic iterations.
An example of use is finally provided in the field of
hash functions.

The remainder of this research work, which summarizes 
our recent discoveries in truly chaotic finite machines
and theirs applications, is organized as follows.

\section{The Mathematical Theory of Chaos}
In the whole document, to prevent from any conflicts and
to avoid unreadable writings, we have considered the following notations,
usually in use in discrete mathematics:
\begin{itemize}
\item The $n-$th term of the sequence $s$ is denoted by $s^n$.
\item The $i-$th component of vector $v$ is $v_i$.
\item The $k-$th composition of function $f$ is denoted by $f^{k}$. Thus $f^{k} = f \circ f \circ \hdots \circ f$, $k$ times.
\item The derivative of $f$ is $f'$.
\end{itemize}
$\mathcal{P}(X)$ is the set of subsets of $X$.
On the other hand $\mathds{B}$ stands for the set $\{0;1\}$ with its
usual algebraic structure (Boolean addition, multiplication, and negation),
while  $\mathds{N}$ and
$\mathds{R}$ are the usual notations of the following respective sets:
natural numbers and real numbers.
$\mathcal{X}^\mathcal{Y}$ is the set of applications 
from $\mathcal{Y}$ to $\mathcal{X}$, and thus $\mathcal{X}^\mathds{N}$
means the set of sequences belonging in $\mathcal{X}$.
We will use the notation $\lfloor x \rfloor$ for the integral part of a real $x$,
that is, the greatest integer lower than $x$.
Finally, $\llbracket a; b \rrbracket = \{a, a+1, \hdots, b\}$ is
the set of integers between  $a$ and $b$.

In this section are presented various understanding of a chaotic behavior for a discrete
dynamical system.

\subsection{Approaches similar to Devaney}

In these approaches, three ingredients are required for unpredictability~\cite{Formenti1998}. 
Firstly, the system must be intrinsically complicated, undecomposable: it cannot be simplified into two
subsystems that do not interact, making any divide and conquer strategy 
applied to the system inefficient. In particular, a lot of orbits must visit
the whole space. Secondly, an element of regularity is added, to counteract
the effects of the first ingredient, leading to the fact that closed points
can behave in a completely different manner, and this behavior cannot be predicted.
Finally, sensibility of the system is demanded as a third ingredient, making that
close points can finally become distant during iterations of the system.
This last requirement is, indeed, often implied by the two first ingredients.
Having this understanding of an unpredictable dynamical system, Devaney has
formalized in~\cite{Devaney} the following definition of chaos.

\begin{definition}
A discrete dynamical system $x^0 \in \mathcal{X}, x^{n+1}=f(x^n)$ on a
metric space $(\mathcal{X},d)$ is chaotic according to Devaney
if:
    \begin{enumerate}
\item \emph{Transitivity:} For each couple of open sets $A,B \subset \mathcal{X}$, there exists $\exists k \in \mathbb{N}$ such that $f^k(A)\cap B \neq \varnothing$.
\item \emph{Regularity:} Periodic points are dense in $\mathcal{X}$.
\item \emph{Sensibility to the initial conditions:} There exists $\varepsilon>0$ such that $$\forall x \in \mathcal{X}, \forall \delta >0, \exists y \in \mathcal{X}, \exists n \in \mathbb{N}, d(x,y)<\delta \textrm{ and } d(f^n(x),f^n(y)) \geqslant \varepsilon.$$
\end{enumerate}
\end{definition}

The system can be intrinsically complicated for various other understanding of this wish, that are 
not equivalent one another, like:
\begin{itemize}
    \item \emph{Undecomposable}: it is not the union of two nonempty closed subsets that are positively invariant ($f(A) \subset A$).
  \item \emph{Total transitivity}: $\forall n \geqslant 1$, the function composition $f^n$ is transitive.
  \item \emph{Strong transitivity}: $\forall x,y \in \mathcal{X},$ $\forall r>0,$ $\exists z \in B(x,r),$ $\exists n \in \mathbb{N},$ $f^n(z)=y.$
  \item \emph{Topological mixing}:  for all pairs of disjoint open nonempty sets $U$ and $V$, there exists $n_0 \in \mathbb{N}$ such that $\forall n \geqslant n_0, f^n(U) \cap V \neq \varnothing$.
\end{itemize}

Concerning the ingredient of sensibility, it can be reformulated as follows.
\begin{itemize}
  \item $(\mathcal{X},f)$ is \emph{unstable} is all its points are unstable: $\forall x \in \mathcal{X},$ $\exists \varepsilon >0,$ $\forall \delta > 0,$ $\exists y \in \mathcal{X},$ $\exists n \in \mathbb{N},$ $d(x,y)<\delta$ and $d(f^n(x),f^n(y)) \geqslant \varepsilon$.
  \item $(\mathcal{X},f)$ is \emph{expansive} is $\exists \varepsilon >0,$ $\forall x \neq y,$ $\exists n \in \mathbb{N},$ $d(f^n(x),f^n(y)) \geqslant \varepsilon$
\end{itemize}

These varieties of definitions lead to various notions of chaos. For instance, 
a dynamical system is chaotic according to Wiggins if it is transitive and
sensible to the initial conditions. It is said chaotic according to Knudsen
if it has a dense orbit while being sensible. Finally, we speak about
expansive chaos when the properties of transitivity, regularity, and expansiveness
are satisfied.

\subsection{Li-Yorke approach}

The approach for chaos presented in the previous section, considering that
a chaotic system is a system intrinsically complicated (undecomposable),
with possibly an element of regularity and/or sensibility, has been
completed by other understanding of chaos. Indeed, as ``randomness''
or ``infinity'', a single universal definition of chaos cannot
be found. The kind of behaviors that are attempted to be described are
too much complicated to enter into only one definition. Instead, a 
large panel of mathematical descriptions have been proposed these last
decades, being all theoretically justified. Each of these definitions
give illustration to some particular aspects of a chaotic behavior.

The first of these parallel approaches can be found in the pioneer
work of Li and Yorke~\cite{Li75}. In their well-known article entitled
``Period three implies chaos'', they rediscovered a weaker formulation of 
the Sarkovskii's theorem, meaning that when a discrete dynamical system 
$(f,[0,1])$, with $f$ continuous, has a 3-cycle, then it has too a 
$n-$cycle, $\forall n \leqslant 2$. The community has not adopted this
definition of chaos, as several degenerated systems satisfy this property.
However, on their article~\cite{Li75}, Li and Yorke have studied too 
another interesting property, which has led to a notion of chaos 
``according to Li and Yorke'' recalled below.

\begin{definition}
Let $(\mathcal{X},d)$ a metric space and $f:\mathcal{X} \longrightarrow
\mathcal{X}$ a continuous map. $(x,y)\in \mathcal{X}^2$ is a scrambled 
couple of points if $\liminf_{n\rightarrow \infty} d(f^n(x),f^n(y))=0$
and  $\limsup_{n\rightarrow \infty} d(f^n(x),f^n(y))>0$: the two orbits oscillate.

A scrambled set is a set in which any couple of points are
a scrambled couple, whereas a Li-Yorke chaotic system is a system 
possessing an uncountable scrambled set.
\end{definition}

\subsection{Topological entropy approach}

Let $f:\mathcal{X} \longrightarrow \mathcal{X}$ be a continuous map on
a compact metric space $(\mathcal{X},d)$. For each natural 
number $n$, a new metric $d_n$ is defined on $\mathcal{X}$ by 
$$d_n(x,y)=\max\{d(f^i(x),f^i(y)): 0\leq i<n\}.$$

Given any $\varepsilon >0$ and $n \geqslant 1$, two points of 
$\mathcal{X}$ are $\varepsilon$-close with respect to 
this metric if their first $n$ iterates are $\varepsilon$-close. This 
metric allows one to distinguish in a neighborhood of an orbit the 
points that move  away from each other during the iteration from the 
points that travel together. 
A subset $E$ of $\mathcal{X}$ is said to be $(n,\varepsilon)$-separated 
if each pair of distinct points of $E$ is at least $\varepsilon$ apart 
in the metric $d_n$. Denote by $N(n, \varepsilon)$ the 
maximum cardinality of a $(n,\varepsilon)$-separated set. 
$N(n, \varepsilon)$ represents the number of distinguishable orbit 
segments of length $n$, assuming that we cannot distinguish points 
within $\varepsilon$ of one another. 

\begin{definition}
The topological 
entropy of the map $f$ is defined by
$$h(f)=\lim_{\epsilon\to 0} \left(\limsup_{n\to \infty} \frac{1}{n}
\log N(n,\epsilon)\right).$$
\end{definition}

The limit defining $h(f)$ may 
be interpreted as the measure of the average exponential growth of the 
number of distinguishable orbit segments. In this sense, it measures 
complexity of the topological dynamical system $(\mathcal{X}, f)$.

\subsection{The Lyapunov exponent}

The last measure of chaos that will be regarded in this document is the Lyapunov exponent. This
quantity characterizes the rate of separation of infinitesimally close 
trajectories. Indeed, two trajectories in phase space with initial 
separation $\delta$ diverge at a rate approximately
equal to $\delta e^{\lambda t}$,
where $\lambda$ is the Lyapunov exponent, which is defined by:

\begin{definition}
Let $f:\mathds{R} \longrightarrow \mathds{R}$ be a differentiable
function, and $x^0\in \mathds{R}$. The Lyapunov exponent is given by
$\lambda(x^0) = \displaystyle{\lim_{n \to +\infty} \dfrac{1}{n} \sum_{i=1}^n \ln \left| ~f'\left(x^{i-1}\right)\right|}.$
\end{definition}

Obviously, this exponent must be positive to have a multiplication of the initial
errors by an exponentially increasing factor, and thus chaos in this understanding.

\section{The So-called Chaotic Iterations}

Our proposal in creating chaotic finite machines is to take a new input 
at each iteration. This process can be realized using a tool called
chaotic iterations.

\subsection{Introducing chaotic iterations}

\begin{definition}
\label{defIC}
Let $f: \mathds{B}^\mathsf{N} \longrightarrow \mathds{B}^\mathsf{N}$ and 
$S \in \mathcal{P} \left(\llbracket1,\mathsf{N}\rrbracket\right)^\mathds{N}$. 
\emph{Chaotic iterations} $(f, (x^0, S))$ are defined by:
$$\left\{
\begin{array}{l}
x^0 \in \mathds{B}^\mathsf{N} \\
\forall n \in \mathds{N}^*, \forall i \in \llbracket 1; \mathsf{N} \rrbracket,  x^{n}_i = \left\{
\begin{array}{ll}
x^{n-1}_{i} & \textrm{ if } i \notin S^n\\
f(x^{n-1})_{i} & \textrm{ if } i \in S^n
\end{array}
\right.
\end{array}
\right.$$
\end{definition}

\emph{A priori}, there is no relation between these chaotic iterations
and the mathematical theory of chaos recalled in the previous section.
On our side, we have regarded whether these chaotic iterations can 
behave chaotically, as it is defined for instance by Devaney, and if so, 
in which application context this behavior can be profitable.
To do so, 
chaotic iterations have first been rewritten as simple discrete 
dynamical systems, as follows.

\subsection{Chaotic Iterations as Dynamical Systems}

To realize the junction between the two frameworks presented previously, the following
material can be introduced~\cite{bgw09:ip}:
\begin{itemize}
\item the shift function: $\sigma : \mathcal{S} \longrightarrow \mathcal{S}, (S^n)_{n \in \mathds{N}} \mapsto (S^{n+1})_{n \in \mathds{N}}$.
\item the initial function, defined by $i : \mathcal{S} \longrightarrow \llbracket 1 ; \mathsf{N} \rrbracket, (S^n)_{n \in \mathds{N}} \mapsto S^0$
\item and $F_f : \llbracket 1 ; \mathsf{N} \rrbracket \times \mathds{B}^\mathsf{N} \longrightarrow \mathds{B}^\mathsf{N},$ $$(k,E) \longmapsto  \left( E_j.\delta(k,j) + f(E)_k.\overline{\delta (k,j)} \right)_{j \in \llbracket 1 ; \mathsf{N} \rrbracket}$$
\end{itemize}
where $\delta$ is the discrete metric.

Let $\mathcal{X} = \llbracket 1 ; \mathsf{N} \rrbracket^\mathds{N} \times \mathds{B}^\mathsf{N},$ and $G_f\left(S,E\right) = \left(\sigma(S), F_f(i(S),E)\right).$ 
Chaotic iterations $\left(f, (S,x^0)\right)$ can be modeled by the 
discrete dynamical system:
$$\left\{
\begin{array}{l}
X^0 = (S,x^0) \in \mathcal{X}, \\
\forall k \in \mathds{N}, X^{k+1} = G_f(X^k).
\end{array}
\right.$$
Their topological disorder can then be studied.
To do so, a relevant distance must be defined on $\mathcal{X}$, as
follows~\cite{gb11:bc}:
$$d((S,E);(\check{S};\check{E})) = d_e(E,\check{E}) +  d_s(S,\check{S})$$
\noindent where $\displaystyle{d_e(E,\check{E}) = \sum_{k=1}^\mathsf{N} \delta (E_k, \check{E}_k)}$, ~~and~ $\displaystyle{d_s(S,\check{S}) = \dfrac{9}{\textsf{N}} \sum_{k = 1}^\infty \dfrac{|S^k-\check{S}^k|}{10^k}}$.

This new distance has been introduced to satisfy the following requirements.
\begin{itemize}
\item When the number of different cells between two systems is increasing, then their distance should increase too.
\item In addition, if two systems present the same cells and their respective strategies start with the same terms, then the distance between these two points must be small because the evolution of the two systems will be the same for a while. Indeed, the two dynamical systems start with the same initial condition, use the same update function, and as strategies are the same for a while, then components that are updated are the same too.
\end{itemize}
The distance presented above follows these recommendations. Indeed, if the floor value $\lfloor d(X,Y)\rfloor $ is equal to $n$, then the systems $E, \check{E}$ differ in $n$ cells. In addition, $d(X,Y) - \lfloor d(X,Y) \rfloor $ is a measure of the differences between strategies $S$ and $\check{S}$. More precisely, this floating part is less than $10^{-k}$ if and only if the first $k$ terms of the two strategies are equal. Moreover, if the $k^{th}$ digit is nonzero, then the $k^{th}$ terms of the two strategies are different.
It can then be stated that
\begin{proposition}
$G_f : (\mathcal{X},d) \to  (\mathcal{X},d)$ is a continuous function
\end{proposition}

With all this material, the study of chaotic iterations as a discrete
dynamical system has then be realized. 
The topological space on which chaotic iterations are defined has
firstly been investigated, leading to the following result~\cite{gb11:bc}:
\begin{proposition}
$\mathcal{X}$ is an infinitely countable metric space, being both
compact, complete, and perfect (each point is an accumulation point).
\end{proposition}
These properties are required in some topological specific 
formalization of a chaotic dynamical system, justifying their
proofs.
Concerning $G_{f_0}$, it has been stated that~\cite{gb11:bc}.
\begin{proposition}
$G_{f_0}$ is surjective, but not injective, and so the dynamical system $(\mathcal{X},G_{f_0})$ is not reversible.
\end{proposition} 
It is now possible to recall the topological behavior of chaotic iterations.

\subsection{The Study of Iterative Systems}

We have firstly stated that~\cite{gb11:bc}:
\begin{theorem}
    $G_{f_0}$ is regular and transitive on $(\mathcal{X},d)$, thus it is 
    chaotic according to Devaney. 
Furthermore, its constant of sensibility is greater than $\mathsf{N}-1$.
\end{theorem}

Thus the set $\mathcal{C}$ of functions $f:\mathds{B}^\mathsf{N} 
\longrightarrow \mathds{B}^\mathsf{N}$ making the chaotic iterations of 
Definition~\ref{defIC} a case of chaos according to Devaney, is a nonempty
set. To characterize functions of $\mathcal{C}$, we have firstly stated
that transitivity implies regularity for these particular iterated 
systems~\cite{bcgr11:ip}. To achieve characterization, we then have introduced the
following graph.

\begin{figure}
    \centering
   \includegraphics[scale=0.55]{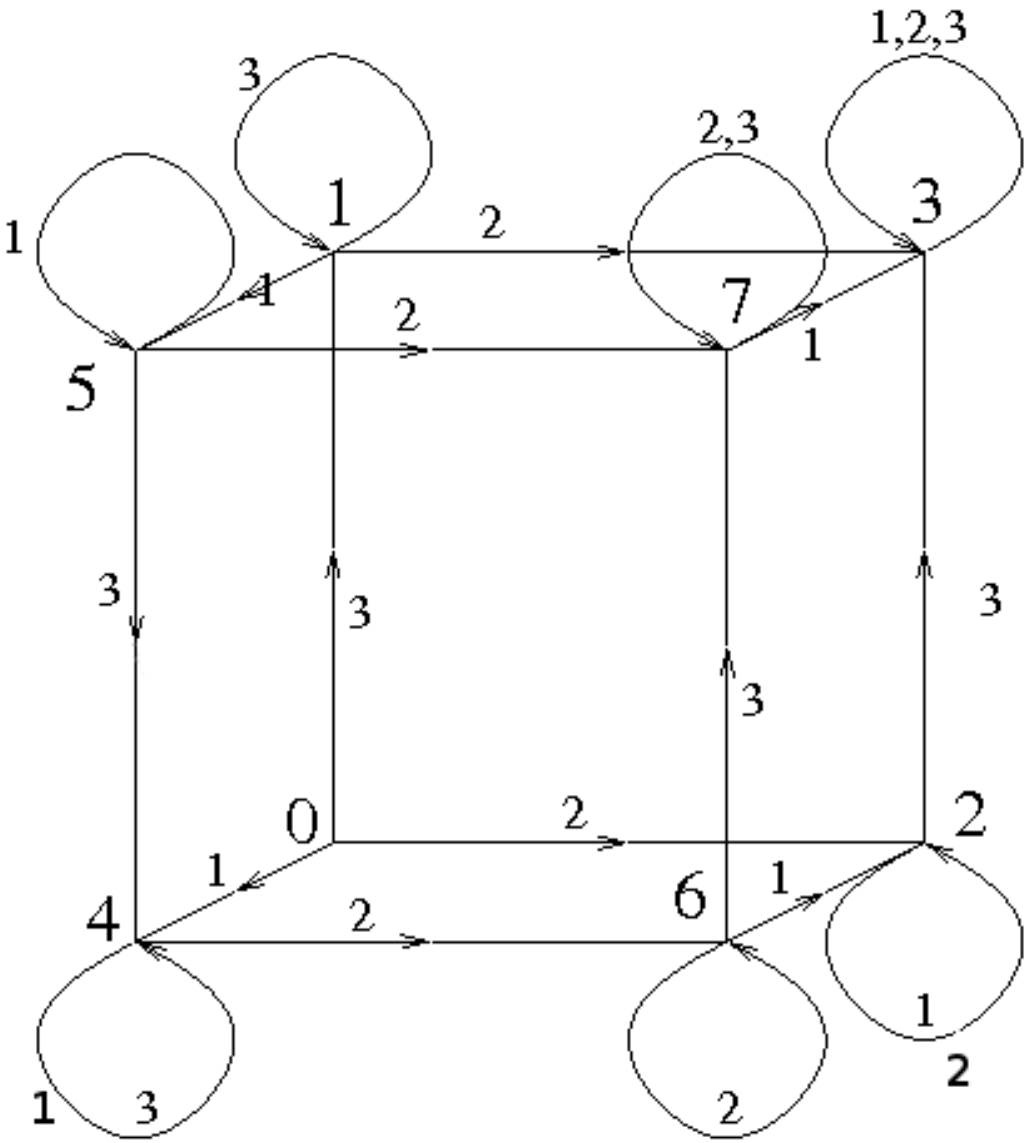}
   \caption{Example of an asynchronous iteration graph}
   \label{GTPIC}
   \end{figure}

Let $f$ be a map from $\mathds{B}^\mathsf{N}$ to itself. The
{\emph{asynchronous iteration graph}} associated with $f$ is the
directed graph $\Gamma(f)$ defined by: the set of vertices is
$\mathds{B}^\mathsf{N}$; for all $x\in\mathds{B}^\mathsf{N}$ and 
$i\in \llbracket1;\mathsf{N}\rrbracket$,
the graph $\Gamma(f)$ contains an arc from $x$ to $F_f(i,x)$. 
The relation between $\Gamma(f)$ and $G_f$ is clear: there exists a
path from $x$ to $x'$ in $\Gamma(f)$ if and only if there exists a
strategy $s$ such that the parallel iteration of $G_f$ from the
initial point $(s,x)$ reaches the point $x'$. Figure~\ref{GTPIC} presents
such an asynchronous iteration graph.
We thus have proven that~\cite{bcgr11:ip}.

\begin{theorem} 
$G_f$  is transitive, and thus chaotic according to Devaney, 
if  and only if $\Gamma(f)$ is strongly connected.
\end{theorem}

This characterization makes it possible to quantify the number of 
functions in $\mathcal{C}$: it is equal to
 $\left(2^\mathsf{N}\right)^{2^\mathsf{N}}$.
Then the study of the topological properties of disorder of these
iterative systems has been further investigated, leading to the following
results.

\begin{theorem}
 $\forall f \in \mathcal{C}$, $Per\left(G_f\right)$ is infinitely countable, $G_f$ is strongly transitive and is chaotic according to Knudsen. It is thus undecomposable, unstable, and chaotic as defined by Wiggins.
 \end{theorem}

 \begin{theorem}
$\left(\mathcal{X}, G_{f_0}\right)$ is topologically mixing, 
     expansive (with a constant equal to 1), chaotic as defined by 
     Li and Yorke, and has a topological entropy and an exponent of Lyapunov 
     both equal to $ln(\mathsf{N})$.
\end{theorem}

At this stage, a new kind of iterative systems that only manipulates 
integers have been discovered, leading to the questioning of their 
computing for security applications. In order to do so, the possibility
of their computation without any loss of chaotic properties has first 
been investigated. These chaotic machines are presented in the next
section.


\section{Chaotic Turing Machines}

\subsection{General presentation}

Let us consider a given algorithm. Because it must be 
computed one day, it is always possible to translate it as a Turing 
machine, and this last machine can be written as $x^{n+1} = f(x^n)$ in 
the following way. Let $(w,i,q)$ be the current configuration of the 
Turing machine (Figure~\ref{Turing}), where $w=\sharp^{-\omega} w(0) 
\hdots w(k)\sharp^{\omega}$ is the paper tape, $i$ is the position of 
the tape head, $q$ is used for the state of the machine, and $\delta$ is 
its transition function (the notations used here are well-known and 
widely used). We define $f$ by:
\begin{itemize}
\item $f(w(0) \hdots w(k),i,q) = ( w(0) \hdots w(i-1)aw(i+1)w(k),i+1,q')$, if  $\delta(q,w(i)) = (q',a,\rightarrow)$,
\item $f( w(0) \hdots w(k),i,q) = (w(0) \hdots w(i-1)aw(i+1)w(k),i-1,q')$,  if $\delta(q,w(i)) = (q',a,\leftarrow)$.
\end{itemize}
Thus the Turing machine can be written as an iterate function 
$x^{n+1}=f(x^n)$ on a well-defined set $\mathcal{X}$, with $x^0$ as the 
initial configuration of the machine. We denote by $\mathcal{T}(S)$ the 
iterative process of the algorithm $S$.

\begin{figure}
  \centering
\includegraphics[scale=0.5]{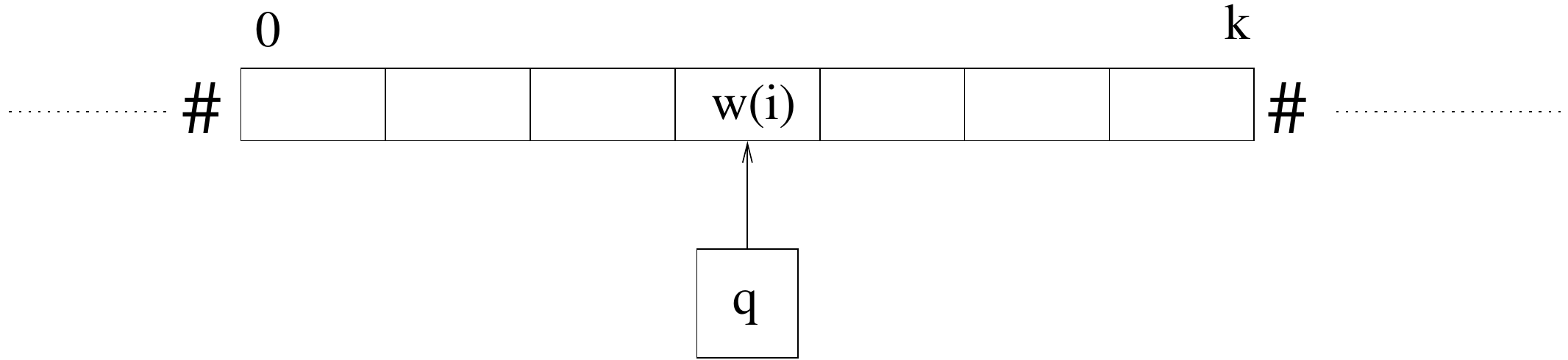}
\caption{Turing Machine}
\label{Turing}
\end{figure}

Let $\tau$ be a topology on $\mathcal{X}$. So the behavior of this 
dynamical system can be studied to know whether or not the algorithm
is $\tau$-chaotic. 
Let us now explain how it is possible to have true chaos in a finite state machine.

\subsection{Practical Issues}

\label{practical}
\label{Chaosincomputer}

%
%
%
Up to now, most of computer programs presented as chaotic lose 
their chaotic properties while computing in the finite set of machine numbers.
The algorithms that have been presented as chaotic usually act as follows.
After having received its initial state, the machine works alone with no interaction with the outside world.
Its outputs only depend on the different states of the machine.
The main problem which prevents speaking about chaos in this particular situation is that when a finite state machine reaches the same internal state twice, the two future evolution are identical.
Such a machine always finishes by entering into a cycle while iterating.
This highly predictable behavior cannot be set as chaotic, at least as expressed by Devaney.
Some attempts to define a discrete notion of chaos have been proposed, but 
they are not completely satisfactory and are less recognized than the notions exposed in a previous section.

The next stage was then to prove that chaos is possible in finite machine.
The two main problems are that: (1) Chaotic sequences are usually 
defined in the real line whereas define real numbers on computers 
is impossible. (2) All finite state machines always enter into a cycle
when iterating, and this periodic behavior cannot be stated as chaotic.

The first problem is disputable, as the shadow lemma proves that, when
considering the sequence $x^{n+1} = trunc_k\left(f(x^n)\right)$, where
$(f,[0,1])$ is a chaotic dynamical system and $trunc_k(x) = 
\dfrac{\lfloor 10^k x \rfloor}{10^k}$ is the truncated version of 
$x\in \mathds{R}$ at its $k-$th digits, then the sequence $(x^n)$ is as
close as possible to a real chaotic orbit. Thus iterating a chaotic 
function on floating point numbers does not deflate the chaotic behavior
as much. However, even if this first claim is not really a problem, we
have prevent from any disputation by considering a tool (the chaotic 
iterations) that only manipulates integers bounded by $\mathsf{N}$.

The second claim is surprisingly never considered as an issue when
considering the generation of randomness on computers.
However, the stated problem can be solved in the following way.
The computer must generate an output $O$ computed from its current state $E$ \emph{and} the current value of an input $S$, which changes at each iteration (Figure~\ref{fig:Mealy}).
Therefore, it is possible that the machine presents the same state twice, but with two future evolution completely different, depending on the values of the input.
By doing so, we thus obtain a machine with a finite number of states, which can evolve in infinitely different ways, due to the new values provided by the input at each iteration.
Thus such a machine can behave chaotically.

\begin{figure}
\centerline{\includegraphics[scale=0.5]{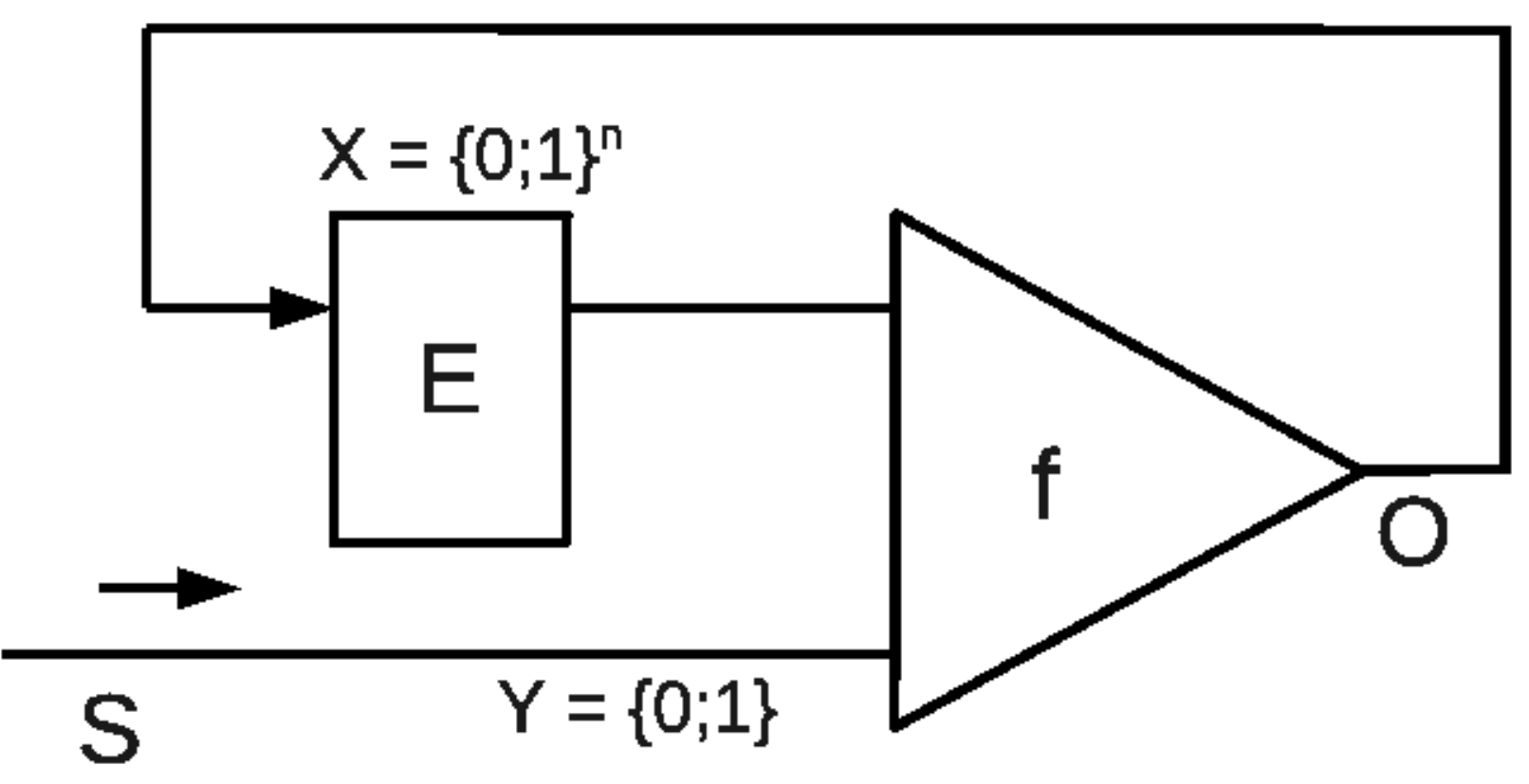}}
\caption{A chaotic finite-state machine. At each iteration, a new value is taken from the outside world (S). It is used by f as input together with the current state (E).}
\label{fig:Mealy}
\end{figure}

\section{Application to Hash Functions}

\subsection{Definitions}

This section is devoted to a concrete realization of such
a finite state chaotic machine in the computer science 
security field. We will show that, given a secured hash 
function, it is possible to realize a post-treatment on
the obtained digest using chaotic iterations that 
preserves the security of the hash function. Furthermore,
if the media to hash is obtained frame by frame from a stream,
the resulted hash machine inherits the chaos properties of
the chaotic iterations presented previously.
For the interest to add chaos properties to an hash function,
among other things regarding their diffusion and confusion~\cite{Shannon49},
reader is referred to the following experimental studies:~\cite{bcg11:ip,gb11:bc,bcg12:ij}.

Let us firstly introduce some definitions.

\begin{definition}[Keyed One-Way Hash Function]
Let $\Gamma$ and $\Sigma$ be two alphabets,   
let $k \in K$ be a key in a given key space,
let $l$ be a natural numbers which is the length of the output message,
and let $h : K \times  \Gamma^{+} \rightarrow \Sigma^{l}$ be a function that associates 
a message in $\Sigma^{l}$ for each pair of key, word in  
$K \times  \Gamma^{+}$.
The set of all functions $h$ is partitioned into classes
of functions $\{h_k : k \in K \}$
indexed by a key $k$ and such that 
$h_{k}: \Gamma^{+} \rightarrow \Sigma^{l}$ is defined by
$h_{k}(m) = h(k,m)$, \textit{i.e.}, $h_{k}$ generates a message digest of length $l$. 
\end{definition}

\begin{definition}[Collision resistance]
For a keyed hash function $h:\mathds{B}^k\times \mathds{B}^* \longrightarrow \mathds{B}^n$, define the advantage of an adversary $\mathsf{A}$ for finding a collision as
\begin{equation}
Adv_\mathsf{A} = Pr \left[\begin{array}{c} K \xleftarrow{\$} \mathds{B}^k \\ (m,m') \leftarrow \mathsf{A}(K)\end{array} : \begin{array}{c} m \neq m' \\ h(K,m)=h(K,m')\end{array}\right]
\end{equation}
where $\$$ means that the element is pick randomly. The insecurity of $h$ with respect to collision resistance is
\begin{equation}
InSec_h(t) = \max_\mathsf{A}\left\{Adv_\mathsf{A}\right\}
\end{equation}
when the maximum is taken over all adversaries $\mathsf{A}$ with total running time $t$.
\end{definition}

In other words, an adversary should not be able to find a collision, that is, two distinct messages $m$ and $m'$ such that $h(m) = h(m')$.

\begin{definition}[Second-Preimage Resistance]
For a keyed hash function $h:\mathds{B}^k\times \mathds{B}^* \longrightarrow \mathds{B}^n$, define the advantage of an adversary $\mathsf{A}$ for finding a second-preimage as
\begin{equation}
Adv_\mathsf{A}(m) = Pr \left[ \begin{array}{c}K \xleftarrow{\$} \mathds{B}^k \\ m' \xleftarrow{\$} \mathsf{A}(K)\end{array} : \begin{array}{c} m \neq m' \\ h(K,m)=h(K,m')\end{array}\right]
\end{equation}
The insecurity of $h$ with respect to collision resistance is
\begin{equation}
InSec_h(t) = \max_\mathsf{A}\left\{\max_{m\in\mathds{B}^k}\left\{Adv_\mathsf{A}(m)\right\}\right\}
\end{equation}
when the maximum is taken over all adversaries $\mathsf{A}$ with total running time $t$.
\end{definition}

That is to say, an adversary given a message $m$ should not be able to find another message $m'$ such that $m\neq m'$ and $h(m)=h(m')$.
Let us now give a post-operative mode that can be applied to a cryptographically secure hash function without loosing 
the cryptographic properties recalled above.


\begin{definition}
Let 
\begin{itemize}
\item $k_1,k_2,n \in \mathds{N}^*$,
\item $h:(k,m) \in \mathds{B}^{k_1}\times\mathds{B}^* \longmapsto h(k,m) \in \mathds{B}^n$ a keyed hash function,
\item $S:k\in \mathds{B}^{k_2} \longmapsto \left(S(k)^i\right)_{i\in \mathds{N}} \in \llbracket 1,n\rrbracket^\mathds{N}$:
\begin{itemize}
\item either a cryptographically secure pseudorandom number generator (PRNG), 
\item or, in case of a binary input stream $m = m^0 || m^1 ||  m^1 || \hdots$ 
 where $\forall i, |m^i| = n$, $\left(S(k)^i\right)_{i\in \mathds{N}} = \left(m^k\right)_{i\in \mathds{N}}$.
%
\end{itemize}
\item $\mathcal{K}=\mathds{B}^{k_1}\times\mathds{B}^{k_2}\times \mathds{N}$ called the \emph{key space}, 
\item and $f:\mathds{B}^n \longrightarrow \mathds{B}^n$ a bijective map.
\end{itemize}
We define the keyed hash function $\mathcal{H}_h:\mathcal{K}\times\mathds{B}^* \longrightarrow \mathds{B}^n$ by the following procedure\\
\begin{tabular}{ll}
\underline{\textbf{Inputs:}} & $k = (k_1,k_2,n)\in \mathcal{K}$\\
                            & $m \in \mathds{B}^*$\\
\underline{\textbf{Runs:}} & $X=h(k_1,m)$, or $X=h(k_1,m^0)$ if $m$ is a stream\\
                           & for $i=1, \hdots, n:$\\
                           & ~~~~ $X=G_f(X,S^i)$\\
                           & return $X$
\end{tabular}
\end{definition}
%
%
%
%
%
%
%
%
%
%
%

$\mathcal{H}_h$ is thus a chaotic iteration based post-treatment 
on the inputted hash function $h$.
The strategy is provided by a secured PRNG when the machine operates
in a vacuum whereas it is redetermined at each iteration from the
input stream in case of a finite machine open to the outside.
By doing so, we obtain a new hash function $\mathcal{H}_h$ with $h$, 
and this new one has a chaotic dependence regarding the inputted stream.

\subsection{Security proofs}

The two following lemma are obvious.

\begin{lemma}
If $f: \mathds{B}^n \longrightarrow \mathds{B}^n$ is bijective, then $\forall S \in \llbracket 1,n \rrbracket$, the map $G_{f,S}:x \in \mathds{B}^n \rightarrow G_f(x,S)_1 \in \mathds{B}^n$ is bijective too.
\end{lemma}

\begin{proof}
Let $y=(y_1, \hdots, y_n) \in \mathds{B}^n$ and $S\in \llbracket 1,n \rrbracket$.
Thus $$G_{f,S}(y_1,\hdots,y_{S-1},f^{-1}(y_S),y_{S+1},\hdots, y_n)_1=y.$$
So $G_{f,S}$ is a surjective map between two finite sets.
\end{proof}

\begin{lemma}
Let $S \in \llbracket 1,n \rrbracket^\mathds{N}$ and $N \in \mathds{N}^*$. If $f$ is bijective, then $G_{f,S,N}:x \in \mathds{B}^n \longmapsto G_f^N(x,S)_1 \in \mathds{B}^n$ is bijective too.
\end{lemma}

\begin{proof}
Indeed, $G_{s,f,n}=G_{f,S^n}\circ \hdots \circ G_{f,S^0}$ is bijective as a composition of bijective maps.
\end{proof}

We can now state that,

\begin{theorem}
If $h$ satisfies the collision resistance property, then it is the case too for $\mathcal{H}_h$.
And if $h$ satisfies the second-preimage resistance property, then it is the case too for $\mathcal{H}_h$.
\end{theorem}

\begin{proof}
Let $A(k_1,k_2,n)=(m_1,m_2)$ such that $\mathcal{H}_h\left((k_1,k_2,n),m_1\right) = \mathcal{H}_h\left((k_1,k_2,n),m_2\right)$. Then $G_{f,S(k_2),n}\left(h(m_1)\right) = G_{f,S(k_2),n}\left(h(m_2)\right)$. So $h(m_1,k_1)=h(m_2,k_1)$.

For the second-preimage resistance property, let $m,k \in \mathds{B}^*\times\mathcal{K}$. If a message $m'\in\mathds{B}^*$ can be found such that $\mathcal{H}_h(k,m)=\mathcal{H}_h(k,m')$, then $h(k_1,m)=h(k_1,m')$: a second-preimage for $h$ has thus be found.
\end{proof}

Finally, as $\mathcal{H}_h$ simply operates chaotic iterations with strategy $\mathcal{S}$ provided at each iterate by the media, we have:

\begin{theorem}
In case where the strategy $\mathcal{S}$ is the bitwise xor between a secured PRNG and the input stream,
the resulted hash function  $\mathcal{H}_h$ is chaotic.
\end{theorem}

\begin{remark}
$\mathcal{S}$ should be $m^k\oplus x^k$ where $(x^k)$ is provided by a secured PRNG if security of $\mathcal{H}_h$ is required.
\end{remark}

\section{Conclusion}

In this article, the research we have previously done in the field of truly chaotic
finite machines are summarized and clarified to serve as an introduction to our 
approach. This approach consists in considering a specific family of discrete
dynamical systems that iterate on a set having the form 
$\mathcal{X}=\mathcal{P}\left(\llbracket 1, N \rrbracket\right)^\mathds{N}\times \mathds{B}^\mathds{N}$,
making it possible to obtain pure, non degenerated chaos on finite machines. These particular
dynamical systems are called chaotic iterations. Our method consists in considering the left part
of $\mathcal{X}$ as the tape of the Turing machine whereas the right part is the
state register of the machine. Chaos implies here that if the initial tape and the initial state
are not known exactly, then for some transition function the evolution of the iterates of
the Turing machine, or in other words the evolution of the state register and of the tape,
cannot be predicted. We remark too that the initial tape has not to be inputted integrally into
the machine: it can be provided by, for instance, a video stream whose hash value is updated
at each new received frame.

In our previous research papers, we have provided a necessary and sufficient condition on a Moore
machine to behave as chaotic, this condition being the strong connectivity of an associated
large graph. A sufficient condition of chaoticity on a smaller graph has been proven too.
In future work, the authors' intention is to extend these results to the Turing machines, to
determine on which conditions on the transition function such machines have a stochastic 
behavior. Furthermore, only a special kind of Turing machines has been investigated until now, 
and the authors' desire is to extend these results to all possible machines. A concrete 
chaotic machine should then be designed and studied. Finally new applications will be
detailed.

\bibliographystyle{plain}
\bibliography{mabase}

\end{document}